\documentclass[12pt]{article}
\usepackage{graphicx}
\usepackage{amssymb}
\usepackage{amsmath}
\usepackage{amsfonts}
\usepackage{amsthm}
\usepackage[english]{babel}
\usepackage[latin1,ansinew]{inputenc}
\usepackage{a4wide}
\usepackage{color}
\newcommand{\be}{\begin{eqnarray}}
\newcommand{\ee}{\end{eqnarray}}
\newtheorem{theo}{Theorem}
\newtheorem{lemma}{Lemma}
\newtheorem{coro}{Corollary}
\newtheorem{defi}{Definition}

\begin{document}

\title
{Godunov Variables in Relativistic Fluid Dynamics}

\author{\it Heinrich Freist\"uhler \footnotemark[1]} 

\date{June 19, 2017}

\maketitle

\footnotetext[1]{Department of Mathematics, University of Konstanz, 78457 Konstanz, Germany}


\begin{abstract}
This note presents 
Godunov variables and 4-potentials
for the relativistic Euler equations of barotropic fluids. 
The associated additional conservation/production law has different  
interpretations for different fluids. In particular it refers to \emph{entropy} 
in the case of thermobarotropic fluids, and to \emph{matter} in the case of 
isentropic fluids.
The paper also presents an explicit formula for the generating function 
of the Euler equations in the case of ideal gases.
It pursues ideas on symmetric hyperbolicity going back to 
Godunov (cf.\ also Lax and Friedrichs as well as Boillat) 
that were  elaborated as Ruggeri and Strumia's 
theory of convex covariant density systems,
\end{abstract}

\newpage
\section{Convex covariant density systems}
\setcounter{equation}0
In famous papers of 1954, 1957, and 1961,
Friedrichs presented the notion of symmetric hyperbolic systems of first-order
partial differential equations and started their mathematical theory \cite{Fr},
Lax defined the concept of hyperbolic systems of conservation laws \cite{Lx} and started their theory, 
and Godunov identified a class of systems of conservation laws that are symmetric hyperbolic when written 
in appropriate natural variables \cite{Go}. 
All three steps were strongly motivated from  physics, which in each case provides many important examples.
After related contributions by Lax and Friedrichs 
\cite{LxFr} and Boillat \cite{Bo}, Ruggeri and Strumia formulated in 1981 a covariant version of 
first-order symmetric hyperbolicity as a theory of ``convex covariant density systems'' \cite{RS}, 
notably for application exactly to relativistic fluid dynamics.

\begin{defi}\label{GV4P}
One calls $\Upsilon=(\Upsilon_0,\ldots,\Upsilon_N)$ \emph{Godunov variables}
and a vector $X^\beta(\Upsilon)$ a \emph{4-potential} for a covariant causal system of conservation laws
\be\label{claws}
\frac{\partial}{\partial x^\beta} F^{a\beta}=0,\quad a=0,\ldots,N,
\ee
if 
\be\label{Fis}
F^{a\beta}=\frac{\partial X^\beta(\Upsilon)}{\partial \Upsilon_a} 
\ee
and  
\be\label{Hessianpositive}
\left(\frac{\partial^2 X^\beta(\Upsilon)}{\partial \Upsilon_a\Upsilon_g}T_\beta\right)_{a,g=0,\ldots,N}
\quad \text{is definite for all }T_\beta \text{ with }T_\beta T^\beta<0. 
\ee
\end{defi}
\begin{coro}
In the situation of Definition \ref{GV4P}, equations \eqref{claws} can be written 
as a quasilinear symmetric system 
\be
A^{a\beta g}(\Upsilon)\frac{\partial \Upsilon_g}{\partial x^\beta}=0, 
\ee
which has every time-like vector $T_\beta$ as a direction of hyperbolicity.
All smooth solutions of \eqref{claws} satisfy the additional 
conservation law
\be\label{addclaw}
\frac{\partial}{\partial x^\beta}\left(X^\beta(\Upsilon)-F^{a\beta}(\Upsilon)\Upsilon_a\right)=0.
\ee
\end{coro}
Note that symmetry and hyperbolicity follow with
$$
A^{a\beta g}
\equiv 
\frac{\partial^2 X^\beta}{\partial \Upsilon_a\partial\Upsilon_g},  \quad a,g=0,\ldots,N,
$$
and equation \eqref{addclaw} is a consequence of
\be\label{confirmaddclaw}
\frac{\partial X^\beta}{\partial x^\beta}
=
\frac{\partial X^\beta}{\partial \Upsilon_a}\frac{\partial \Upsilon_a}{\partial x^\beta}
=
F^{a\beta}\frac{\partial \Upsilon_a}{\partial x^\beta}
=
\frac{\partial}{\partial x^\beta}\left(F^{a\beta}\Upsilon_a\right).
\ee
The contents of Definition 1 and Corollary 1 closely follow \cite{Fr,Lx,Go,LxFr,Bo,RS}.\footnote{The 
only difference from \cite{RS} is the fact that in consistency with causality, we require the 
definiteness in \eqref{Hessianpositive} (and obtain hyperbolicity)  for \emph{all}\ time-like $T_\beta$  
instead of $\emph{some}$.}

The relativistic dynamics of perfect fluids is governed by the conservation 
equations for energy-momentum, 
\be\label{energmomentumconservation}
{\partial \over \partial x^\beta}\left(T^{\alpha\beta}\right)
=0,
\ee
and matter,
\be\label{matterconservation} 
{\partial \over \partial x^\beta}\left(nU^{\beta}\right)=0,
\ee
where 
\be\label{invstresstensor}
T^{\alpha\beta}=(\rho+p)U^{\alpha}U^{\beta}+p\,g^{\alpha\beta} 
\ee
and $nU^\beta$ are the energy-momentum tensor and the matter current. 
Here $U^\alpha$ denotes the 4-velocity 
of the fluid, and the fluid itself 
is specified by its specific internal energy $e$ as a function of matter  
density $n$ and specific entropy $\sigma$,
from which internal energy $\rho$ and pressure $p$ derive, 
\be\label{generalfluid}
e=e(n,\sigma),\quad
\rho=\rho(n,\sigma)=ne(n,\sigma),\quad
p=p(n,\sigma)=n^2e_n(n,\sigma).
\ee
The five partial differential equations \eqref{energmomentumconservation}, \eqref{matterconservation} 
constitute a system for `five fields': the 4-velocity (three degrees of freedom, as it is 
constrained by unitarity, $U^\alpha U_\alpha=-1$) and two thermodynamic variables (such as $n$ 
and $\sigma$).

Ruggeri and Strumia have shown in \cite{RS} that for very general fluids \eqref{generalfluid},
very similarly to the non-relativistic situation (cf.\ \cite{Go}), the quantities
\be\label{GodVarGen}
\psi_\alpha=\frac{U^\alpha}\theta,\ \alpha=0,1,2,3,
\quad\text{and}\quad 
\psi_4=\frac\mu\theta
\ee
are Godunov variables
for the system 
\eqref{energmomentumconservation}, \eqref{matterconservation},
where 
$$\theta=e_\sigma(n,\sigma)\quad\text{and}\quad \mu=\frac{\rho+p}n-\theta\sigma$$
denote temperature and chemical potential, and the additional  
conservation law\footnote{Cf.\ Section 3 for solutions with shock waves.}
is that for the entropy, 
\be\label{entropyconservation}
{\partial \over \partial x^\beta}\left(n\sigma U^{\beta}\right)=0.
\ee

The purposes of this note are to give related particular treatments for 
(a) barotropic fluids, (b) isentropic fluids, and to (c) derive an explicit formula 
for the 4-potential in the case of ideal gases.

\section{Godunov\! variables\! for\! barotropic\! and\! isentropic\! fluids}
\setcounter{equation}0
A perfect fluid is barotropic and causal if there is a one-to-one relation between 
internal energy and pressure,
\be\label{barotropic}
p=\hat\rho(p),  
\ee
which satisfies
\be\label{SSTL}
\hat\rho'(p)\ge1.
\ee
For barotropic fluids, the conservation laws \eqref{energmomentumconservation}
for energy-momentum are a self-consistent system that determines $\rho$ and $p$ (for given 
initial data) without reference to $n,\theta$ or $s$; this may be called a four-field theory,
as now only one thermodynamic variable (which can be taken to be $\rho$ or $p$) is needed
in addition to $U^\alpha$. It is probably due also to this nice reduction in the number of variables
that barotropic fluids are widely used in physics  (cf.\ \cite{HE,W,C2,ST}).

Remarkably, these four-field theories permit an independent 
analogue of the abovementioned results for five-field theories.
\begin{theo}\label{ThmBaro}
(i) To every causal barotropic fluid \eqref{barotropic},\eqref{SSTL},
there exist essentially unique\footnote{$f$ is unique up to a positive multiplicative constant.
Once $f$ is normalized, say by $f(1)=1$, $X$ is determined up to an additive constant.}
functions f=f(p) and 
$X(\Upsilon)=\hat X(f)$ such that the quantities 
$$
\Upsilon_\alpha=\frac{U_\alpha}f,
$$
are Godunov variables and 
$$
X^\beta(\Upsilon)=
\frac{\partial X(\Upsilon)}{\partial \Upsilon_\beta} 
=
f^3\hat X'(f)\Upsilon^\beta,\quad f=\left(-\Upsilon_\alpha \Upsilon^\alpha\right)^{-1/2} 
$$
is a 4-potential which together represent \eqref{energmomentumconservation} in the
form  \eqref{claws}, \eqref{Fis} with \eqref{Hessianpositive}.\\
(ii) 
The corresponding additional conservation law \eqref{addclaw} is given by 
\be\label{pseudomatterconservation}
{\partial \over \partial x^\beta}\left(\nu U^{\beta}\right)=0,
\ee
with
\be\label{nuis}
\nu=\nu(p)\equiv\frac{\hat\rho(p)+p}{f(p)}=\frac1{f'(p)}.
\ee
\end{theo}
\begin{proof}
(i) Consider a momentarily arbitrary function of the form 
$$
X(\Upsilon)=\hat X(f)\quad\text{with }f=\left(-\Upsilon_\alpha \Upsilon^\alpha\right)^{-1/2}. 
$$
As 
$$
df=f^3\Upsilon^\alpha d\Upsilon_\alpha,
$$
we have 
$$
X^\beta(\Upsilon)=
\frac{\partial X}{\partial \Upsilon_\beta}
=
\pi(f)\Upsilon^\beta\quad\text{with}\quad\pi(f)\equiv f^3\hat X'(f)
$$ 
and 
$$
\frac{\partial^2 X}{\partial \Upsilon_\alpha\partial\Upsilon_\beta}=
f^3\pi'(f)\Upsilon^\alpha\Upsilon^\beta+\pi(f)g^{\alpha\beta}
=
f\pi'(f)U^\alpha U^\beta+\pi(f)g^{\alpha\beta}
$$ 
Now, equation \eqref{Fis}, here
\be\label{T}
T^{\alpha\beta}=\frac{\partial^2 X}{\partial \Upsilon_\alpha\partial\Upsilon_\beta}.
\ee
is equivalent to
$$
\pi(f(p))=p\quad\text{and}\quad f(p)\pi'(f(p))=\hat\rho(p)+p,
$$
and this is equivalent to 
\be\label{odef}
\frac{f'(p)}{f(p)}=\frac1{\hat\rho(p)+p}\quad\text{and}\quad \pi=f^{-1}. 
\ee
From $\pi$, one determines $\hat X$ as
$$
\hat X(f)=\int (1/f)^{3}\pi(f) df.
$$
(ii) 
This follows from 
$$
X^\beta(\Upsilon)-T^{a\beta}(\Upsilon)\Upsilon_a
=
\pi'(f)U^\beta.
$$
\end{proof}
Among the barotropic fluids, \emph{isentropic} fluids are
characterized by the property that their specific internal energy, 
internal energy, and pressure do not depend on an entropy, but solely on the matter density $n$, 
\be\label{eisentropic}
e=e(n)>0,\quad \rho=\rho(n)=ne(n)>0,\quad p=p(n)=n^2e'(n)>0.
\ee
For isentropic fluids, also the specific enthalpy $h$ is a function of $n$ alone,
\be
h=\frac{\rho+p}n\equiv h(n).
\ee
\begin{lemma}
For isentropic fluids \eqref{eisentropic}, the index $f$ can be chosen as
the specific enthalpy, 
\be\label{fish}
f(p(n))=h(n).  
\ee
\end{lemma}
\begin{proof}
From
$$
p'(n)=2ne'(n)+n^2e''(n)=nh'(n)
$$
one obtains
$$
\frac{dh}{dp}=\frac{h'n)}{p'(n)}=\frac1n=\frac h{\rho+p},
$$ 
which means that $f$ defined through $f(p(n))=h(n)$ satisfies the differential equation 
\eqref{odef}$_1$.
\end{proof}

{\bf Remark 1.} The quantity $f$ identified by Theorem \ref{ThmBaro}
is the \emph{index}\ of a barotropic fluid that 
Lichnerowicz has introduced in \cite{Li} as
$$ 
f=f(p)\equiv\exp\int\frac{dp}{\hat\rho(p)+p}.
$$ 
Lichnerowicz worked not with the Godunov variable $U^\alpha/f$, but with the `dynamic' velocity
$fU^\alpha$.

{\bf Remark 2.} 
The integrability of the 4-potential,
here the existence of $X(\Upsilon)=\hat X(f)$ such that 
$X^\beta(\Upsilon)=\partial X(\Upsilon)/\partial \Upsilon_\beta$, can be viewed as a consequence
of the symmetry of the energy-momentum tensor $T^{\alpha\beta}$ \cite{GL}. 
On the other hand, already the naturally needed isotropy  of the mapping $\Upsilon^\beta\mapsto X^\beta$ induces
the form $X^\beta=\tilde X(\Upsilon_\alpha\Upsilon^\alpha)\Upsilon^\beta$, thus the integrability,
and one may conversely re-understand the symmetry of  $T^{\alpha\beta}$ as a consequence of 
the isotropicity requirement for fundamental equations. 

{\bf Remark 3.}
For isentropic fluids, Eq.\ \eqref{odef} reads 
\be\label{piandh}
\pi(h(n))=p(n)
\quad\text{and}\quad 
\pi'(h(n))=n,
\ee
which implies
$$
\rho(n)+\pi(h)=nh,
$$
i.e., density and pressure are Legendre conjugate, with matter density and enthalpy,
$$
n=\pi'(h)\quad\text{and}\quad h=\rho'(n),
$$
as dual variables. 
For instance, the massless ($m=0$) isentropic $\gamma$-law gases $e(n)=\frac1\gamma n^{\gamma-1}$ 
have 
$$
\rho(n)=\frac1\gamma n^\gamma\quad\text{and}\quad \pi(h)=\frac1\delta h^\delta \quad\text{with }
\gamma+\delta=\gamma\delta,
$$
and a beautiful example is given by the limiting `stiff' 
fluid, 
that corresponds to $\gamma=\delta=2$, the fixed point
of the Legendre transform.\footnote{For the stiff fluid, cf.\ \cite{Li,C2,F}.}
\section{The additional law of conservation/production}
\setcounter{equation}0
The purpose of this section is to discuss the interpretation of the additional conservation law
\eqref{pseudomatterconservation}. To illustrate that there is a spectrum of possibilities, 
we begin with two particular cases.
The considerations and results of Section 2 hold in particular
for fluids of the form\footnote{Note that for any barotropic fluid in 
product form $e(n,\sigma)=\check e(n) r(\sigma)$ with non-constant $r$, $\check e(n)$ must be a power law.}
\be\label{productform}
e(n,\sigma)=n^{\gamma-1} r(\sigma),\quad 1<\gamma<2,
\ee
which includes ideal gases of vanishing or negligible particle mass,
corresponding to 
\be\label{mig}
r(\sigma)=k\exp(\sigma/c_v)
\ee
as well as ``double $\gamma$-law'' gases\footnote{These gases have $\rho=kn^\gamma\sigma^\gamma$.
An example is pure radiation, $\gamma=4/3$.},
for which 
\be\label{doublegammalaw}
r(\sigma)=k\sigma^\gamma.
\ee
\begin{lemma}
(i) For double $\gamma$-law gases \eqref{doublegammalaw} the index $f$ is a constant multiple of 
the temperature $\theta$ and the quantity $\nu$ introduced in Theorem \ref{ThmBaro} is a constant 
multiple of the entropy density $n\sigma$; the additional conservation law \eqref{addclaw} 
is that of entropy, 
\eqref{entropyconservation}.
(ii) These three assertions are wrong for the massless ideal gases \eqref{mig}. 
\end{lemma}
\begin{proof}
(i) The product form \eqref{productform}
implies 
\be\label{fnu}
f\sim n^{\gamma-1}\left(r(\sigma)\right)^{1-1/\gamma},
\quad
\nu\sim n (r(\sigma))^{1/\gamma},
\ee
and in case of \eqref{doublegammalaw}
\be\label{sims}
r(\sigma)^{1/\gamma}\sim\sigma
\quad\text{and}\quad
\theta
=e_\sigma(n,\sigma)
\sim
n^{\gamma-1}r(\sigma)^{1-1/\gamma}.
\ee
(ii) Both relations in \eqref{sims} are wrong in case of \eqref{mig}.
\end{proof}
Easy as it is to obtain now, the following general result seems particularly interesting.
\begin{theo}
For isentropic fluids \eqref{eisentropic}, the quantity $\nu$ can be chosen as
the particle number density, 
$$\nu(p(n))=n,$$
and the additional conservation law \eqref{pseudomatterconservation}
is that of matter, \eqref{matterconservation}.
\end{theo}
\begin{proof}
In view of \eqref{fish}, 
equation \eqref{nuis} implies
$$ 
\nu=\frac{\rho+p}h=n.
$$ 
\end{proof}
In other words, isentropic fluids have the property that the conservation of matter,
\eqref{matterconservation},
is \emph{implied}\ by that of energy-momentum, \eqref{energmomentumconservation}.

{\bf Remark 4.}
Hawking and Ellis\footnote{While these authors do not distinguish between `barotropic' and 
`isentropic' (at least  at the time when they wrote \cite{HE}), we here stick to the above 
definitions of the two notions.} point out that for arbitrary barotropic fluids, one can ``introduce''  
a ``conserved quantity'' and an ``internal energy'' (\cite{HE}, p.\ 70, l.\ 6,7 from below,
$\rho$ and $\epsilon$ in their notation). Identical (though not derived there) with our $\nu$,
in particular this conserved quantity is thus \emph{not}\ always the matter 
density\footnote{The authors
didn't claim it was ...  .} $n$, but deviates from it by a generically 
non-constant factor. 

We now turn to the fact that the additional conservation law ``for'' $\nu$ indeed
holds only for smooth solutions of the original system (which here is the four-field theory 
\eqref{energmomentumconservation} by itself). For the five-field context, 
this phenomenon is well-known for the entropy law, which is replaced, 
in the presence of shock waves,
by the inequality
\be\label{entropyproduction}
{\partial \over \partial x^\beta}\left(n\sigma U^{\beta}\right)>0,
\ee
that expresses the \emph{second law of thermodynamics}. (See \cite{I,RS}, in analogy to the non-relativistic 
setting, cf., e.g., \cite{Da}.)

\begin{theo}\label{ThmEntropyproduction}
Assume that for the Euler equations
\eqref{energmomentumconservation},\eqref{invstresstensor}
for a barotropic fluid \eqref{barotropic},\eqref{SSTL},
the acoustic mode is genuinely nonlinear. Then any Lax shock 
solving \eqref{energmomentumconservation}
satisfies the `production law'
\be\label{nuproduction}
{\partial \over \partial x^\beta}
\left(\nu U^\beta\right)>0
\ee
in the sense of distributions, with $\nu$ from  \eqref{nuis}.
\end{theo}
On shock waves in relativistic fluid dynamics,  
some results can be found in 
\cite{T,I,ST,STKochel,FR,FT1}.\footnote{This list is by far not exhaustive.}
For the mathematical notions of shock wave and genuine nonlinearity, see \cite{Lx}. 
For relativistic barotropic fluids, genuine nonlinearity is equivalent to the condition 
\be\label{GNL}
(\rho+\hat p(\rho))\hat p''(\rho)+2(1-\hat p'(\rho))\hat p'(\rho)>0,
\ee
where $\hat p=\hat\rho^{-1}$, 
holding for all $\rho>0$; cf.\ \cite{CB69,Bo73,ST}.

Theorem \ref{ThmEntropyproduction} is an immediate corollary of 
``Statement II'' in Section 6 of \cite{RS}. 
One just has to note that our development in the previous section of the 
present paper implies that what Ruggeri and Strumia denote by $\eta$ is,
in the present case of application, the `production' 
$(\partial/\partial x^\beta)(\nu U^\beta)$ at the shock wave! 

Theorems 2 and 3 readily yield 
\begin{coro}
Assume that for the Euler equations
\eqref{energmomentumconservation},\eqref{invstresstensor}
for a causal isentropic fluid \eqref{eisentropic},\eqref{SSTL},
the acoustic mode is genuinely nonlinear. Then any Lax shock 
solving \eqref{energmomentumconservation}
is accompanied by strictly positive `\emph{matter production}',
\be\label{massproduction}
{\partial \over \partial x^\beta}
\left(n U^\beta\right)>0
\ee
in the sense of distributions.
\end{coro}

{\bf Remark 5.} Shock waves\footnote{though, like here, often understandable from their 
properties as 
ingredients of weak solutions to first-order systems of conservation laws without explicit 
reference to the dissipation mechanism(s) (cf.\ \cite{Da})} are a phenomenon 
of dissipation, and the production laws 
\eqref{entropyproduction},\eqref{nuproduction},\eqref{massproduction}
hold also when the conservation laws \eqref{energmomentumconservation} are augmented
by proper dissipation terms. This will be carried out elsewhere.

\section{Ideal gases}
\setcounter{equation}0
In this section we return to the context originally studied in particular by Ruggeri and Strumia
in \cite{RS} and show an explicit formula for a particular example, namely the ideal gas. 

Assume that $\rho=ne(n,\sigma), p=n^2e_n(n,\sigma)$ where 
\be\label{IGEOS}
e(n,\sigma)=m+kn^{\gamma-1}\exp(\sigma/c_v)
\ee
with $1<\gamma\le 2$, and write $$T^{4\beta}\equiv N^\beta\equiv nU^\beta.$$

\begin{theo}\label{GenIG}
A 
function
$X(\psi_\alpha,\psi)=\hat X(\theta,\psi)$ 
can be explicitly determined such that when written in the variables 
$$
\psi_\alpha=\frac{U_\alpha}\theta \quad\text{and}\quad\psi_4=\psi=\frac{\rho+p}{\theta n}-\sigma, 
$$
the equations 
$$
\frac{\partial}{\partial x^\beta} T^{a\beta}=0,\quad a=0,\ldots,4,
$$
are symmetric hyperbolic, with
\be\label{TIG}
T^{a\beta}=\frac{\partial^2 X}{\partial \psi_a\partial\psi_\beta}.
\ee
\end{theo}
\begin{proof}
Equation \eqref{TIG} yields
$$ 
T^{\alpha\beta}=\theta^3\frac{\partial \hat X(\theta,\psi)}{\partial \theta}\psi^\alpha\psi^\beta
+\hat X(\theta,\psi)g^{\alpha\beta},
\quad
N^\beta=\frac{\partial \hat X(\theta,\psi)}{\partial \psi}\psi^\beta.
$$
The ideal gas equation of state \eqref{IGEOS} leads to 
\begin{eqnarray*}
\hat X(\theta,\psi)
&=&
(\gamma-1)c_v\theta
\left(\frac k{c_v\theta}\right)^\frac 1{1-\gamma}
\exp\left(\frac 1{c_v(1-\gamma)}\frac m\theta\right)
\exp\left(\frac 1{c_v(1-\gamma)}(-\psi{+\gamma c_v})\right)\\
&=&
\tilde k \theta^{1/(1-1/\gamma)}
\exp\left(\frac 1{c_v(\gamma-1)}(\psi-\frac m\theta-\gamma c_v)\right).
\end{eqnarray*}

The symmetric hyperbolicity follows from the fact, mentioned already in Section 2, that 
$\psi_a, a=0,...,4$, are the Godunov variables \cite{RS}.
\end{proof}

\end{document}